\documentclass[times,sort&compress,3p]{elsarticle}
\journal{arXiv}
\usepackage[labelfont=bf]{caption}

\usepackage{graphicx,psfrag,epsf,xcolor}
\usepackage{enumerate}
\usepackage{natbib}
\usepackage{url} 

\DeclareUnicodeCharacter{2113}{l}

\usepackage{lineno}

\usepackage{mathtools}
\usepackage{algorithm}
\usepackage[noend]{algpseudocode}
\usepackage{subcaption}

\usepackage{amsmath,amsfonts,amssymb,amsthm,booktabs,color,epsfig,graphicx,hyperref,url}

\theoremstyle{plain}
\newtheorem{theorem}{Theorem}

\newtheorem{proposition}{Proposition}
\newtheorem{lemma}{Lemma}

\theoremstyle{definition}

\begin{document}

\begin{frontmatter}

\title{Nonparametric estimation of multivariate copula using empirical bayes methods}


\author[1]{Sujit Ghosh}
\author[2]{Lu Lu\corref{mycorrespondingauthor}}

\address[1]{Department of Statistics, North Carolina State University}
\address[2]{Department of Statistics, North Carolina State University}

\cortext[mycorrespondingauthor]{Corresponding author. Email address: \url{llu2@ncsu.edu}}

\begin{abstract}
In the field of finance, insurance, and system reliability, etc., it is often of interest to measure the dependence among variables by modeling a multivariate distribution using a copula. The copula models with parametric assumptions are easy to estimate but can be highly biased when such assumptions are false, while the empirical copulas are non-smooth and often not genuine copula making the inference about dependence challenging in practice. As a compromise, the empirical Bernstein copula provides a smooth estimator but the estimation of tuning parameters remains elusive. In this paper, by using the so-called empirical checkerboard copula we build a hierarchical empirical Bayes model that enables the estimation of a smooth copula function for arbitrary dimensions. The proposed estimator based on the multivariate Bernstein polynomials is itself a genuine copula and the selection of its dimension-varying degrees is data-dependent. We also show that the proposed copula estimator provides a more accurate estimate of several multivariate dependence measures which can be obtained in closed form.
We investigate the asymptotic and finite-sample performance of the proposed estimator and compare it with some nonparametric estimators through simulation studies. An application to portfolio risk management is presented along with a quantification of estimation uncertainty.
\end{abstract}


\end{frontmatter}

\section{Introduction\label{sec:intro}}

Copula models are useful tools for the analysis of multivariate data since by using the well-known Sklar's theorem, any multivariate joint distribution can be decomposed into its univariate marginal distributions and a copula function, which allows for capturing the arbitrary dependence structure between several random variables. As a result, copulas have been widely used in the field of finance, insurance, system reliability, etc. among many other application areas. See, e.g., \citet{jaworski2010copula}, \citet{joe2014dependence} and \citet{nelsen2007introduction} for more details about copulas and their applications. 

Given a random vector $(X_1, \ldots, X_d)$ with joint cumulative distribution function (CDF) $F$ and continuous marginal CDFs $F_j, j \in \{1, \dots, d \}$, by Sklar's theorem (\citet{sklar1959fonctions}), the CDF $F$ can be expressed uniquely as $F(x_1,\ldots, x_d) = C(F_1(x_1),\ldots, F_d(x_d))$, where $C(\cdot)$ denotes the copula function.
A copula is itself the joint CDF of a random vector $(U_1 = F_1(X_1), \ldots,U_d = F_d(X_d)) $ having its marginals as uniform distributions on $[0, 1]$, henceforth denoted by $Unif[0, 1]$. It is to be noted that the original results in \citet{sklar1959fonctions} are also applicable to discrete-valued random variables. However, the focus of this paper is modeling continuous-valued multivariate random vectors. Thus, for the rest of the paper, we assume the marginal CDFs $F_j, j \in \{1, \dots, d \}$ are absolutely continuous.

As copula plays an important role in capturing the general dependence structure between multiple variables, it is critical to estimate copula in an accurate way, especially in higher dimensions where the dependence structure becomes much more complicated, often illusive and may even be supported on a lower-dimensional manifold. One of the primary goals of this paper is to estimate a smooth copula function $C$ from a random sample of $n$ independent identically distributed (iid) observations $(X_{i1}, \ldots, X_{id})\stackrel{iid}{\sim} F(x_1,\ldots, x_d)$ for $i \in \{1,\ldots,n \}$.

Many parametric families have been proposed for modeling multivariate copulas and there has been previous work addressing the corresponding parametric estimation methods. For detailed discussions, see, e.g., \citet{joe2005asymptotic},  \citet{joe2014dependence},  \citet{mcneil2009multivariate},  \citet{nelsen2007introduction}, \citet{smith2011bayesian} and \citet{vzevzula2009multivariate},  etc. 
However, no matter how sophisticated and flexible the parametric models that we may use, they might still lead to biased copula estimates when the parametric model is misspecified and thus may not be able to capture complex dependence structures required in practice. Compared to standard multivariate copulas, vine copula models allow for more flexibility in capturing complex dependency structures using appropriate vine tree structures by choosing bivariate copula families for each node of pair copulas from a vast array of parametric bivariate copulas. But it is often challenging to obtain estimates of multivariate dependence measures which involve high-dimensional integrals that are often algebraically intractable by using vine copulas.

Thus, recognizing some of the above-mentioned limitations of parametric copula models, a variety of nonparametric estimators have been proposed for multivariate copula estimation. Most of the available nonparametric estimators rely on the empirical methods, e.g., the empirical copula and its multilinear extension, the empirical multilinear copula (\citet{deheuvels1979fonction}; \citet{fermanian2004weak}; \citet{genest2017asymptotic}), or kernel-based methods such as local linear estimator (\citet{chen2007nonparametric}), mirror reflection estimator (\citet{gijbels1990estimating}) and improvements of these two estimators (\citet{omelka2009improved}). See also\citet{remillard2009testing} and \citet{scaillet2002nonparametric} for other nonparametric copula estimators. However, except for the empirical multilinear copula, most of these estimators are valid copulas only asymptotically, meaning that they are not necessarily genuine copulas for finite samples. Moreover, multivariate dependent measures (e.g., Spearman's rho, Kendall's tau, etc.) based on such estimated copula could take values outside of their natural range, thus making them unattractive in practice.  On the other hand, there has been recent work on Bayesian nonparametric methods for estimating general d-dimensional copula and among many others, a noteworthy Bayesian nonparametric model is based on an infinite mixture of multivariate Gaussian or the skew-normal copulas proposed by \citet{wu2014bayesian}. The infinite mixture models provide a lot of flexibility in modeling various dependence structures but those typically lack simple (analytic) expressions of dependence measures making them harder to compute in practice. 

The primary focus of this paper is the nonparametric estimation of multivariate copulas for any arbitrary dimensions that are genuine copulas for any finite sample size and are uniformly consistent as the sample size gets large. We consider an extension of the Bernstein copula (\citet{sancetta2004bernstein}), which is a family of copulas defined in terms of multivariate Bernstein polynomials. One of the primary advantages of the Bernstein copula is that it provides a class of nested models that are able to uniformly approximate any multivariate copula with minimal regularity conditions. A simple case of the Bernstein copula is the empirical Bernstein copula, which is a nonparametric copula estimator proposed by \citet{sancetta2004bernstein}. The asymptotic properties of the empirical Bernstein copula are well studied in \citet{janssen2012large} and its application in testing independence is described in \citet{belalia2017testing}. The application of the Bernstein copula to the modeling of dependence structures of non-life insurance risks is provided in \citet{diers2012dependence}, among many other applications.

However, the empirical Bernstein copula has two main drawbacks that could prevent us from obtaining accurate copula estimation for small samples: (i) the empirical Bernstein copula is not necessarily a valid copula itself, which is a common disadvantage for most nonparametric copula estimators; and (ii) the degrees of Bernstein polynomials are often set to be equal to an integer across different dimensions, which limits the flexibility of the Bernstein copula and thus might not be appropriate for large dimensions. 

In order to address the above-described problem (i), \citet{segers2017empirical} show that the empirical Bernstein copula is a genuine copula if and only if all the polynomial degrees are divisors of the sample size, and further propose a new copula estimator called the empirical beta copula, which can be seen as a special case of the empirical Bernstein copula when the degrees of Bernstein polynomials are all set equal to the sample size. The empirical beta copula is a valid copula itself and has been shown to outperform some classical copula estimators in terms of bias and variance. But it always has a larger variance compared to the empirical Bernstein copula with smaller polynomial degrees. It is surprising that much less attention has been given to the problem (ii), and even for equally set degrees, there has been limited work on the data-dependent choice of degrees in the literature. \citet{janssen2012large} recommend an optimal choice of the equal degrees in the bivariate case by minimizing the asymptotic mean squared error. Nevertheless, such a choice requires the knowledge of the first- and second-order partial derivatives, which might not be easy to estimate in practice.  \citet{burda2014copula} put priors on the polynomial degrees, however, their priors don't rely on data or sample size and they use multivariate Bernstein density instead of Bernstein copula density. The Dirichlet process assigned as the prior for the copula doesn't guarantee the copula estimate to be a valid copula itself.  Besides, the number of weights grows exponentially as the dimension increases, leading to computational inefficiency of MCMC methods for larger dimensions. To the best of our knowledge, \citet{lu2020nonparametric} first develop a data-dependent grid search algorithm for the selection of polynomial degrees which has shown superior empirical estimation properties for small to moderate sized samples. But the methodology is limited to bivariate cases and extension to larger dimensions remained challenging.

For the purpose of addressing the two problems as described above, we introduce a new nonparametric smooth estimator for multivariate copula that we call the empirical checkerboard Bernstein copula (ECBC), which is constructed by extending the Bernstein copula allowing for varying degrees of the polynomials. It is shown to be a genuine smooth copula for any number of degrees and any finite sample size. Furthermore, we develop an empirical Bayesian method that takes the data into account to automatically choose the degrees of the proposed estimator using its posterior distribution thereby accounting for the uncertainty of such tuning parameter selection. As shown in  \citet{segers2017empirical}, larger degrees of the Bernstein copula lead to a larger variance of the estimation, so a choice of degrees that is relatively small compared to the sample size but sufficient for a good copula estimation is desirable. The degrees are allowed to be dimension-varying within the Bayesian model, which provides much more flexibility and accuracy, especially in higher dimensions.

It is especially noteworthy that while the focus of the paper is to estimate the copula function, it is straightforward to get a closed-form estimate of the corresponding copula density by taking derivatives of the ECBC. However, direct estimation of a closed-form copula function has many advantages compared to first estimating a copula density, e.g., it is often easier to differentiate than to integrate for higher dimensions. Besides, for those copulas which are not absolutely continuous such as Marshall-Olkin copulas (\citet{embrechts2001modelling}) having support on a possibly lower-dimensional manifold, the direct estimation of the copula density could be difficult. Owing to the closed form of the estimated copula function and its density, it can be shown that the proposed ECBC allows for straightforward estimation of various dependence measures.

The rest of the paper is organized as follows: in Section \ref{sec:meth} we present an empirical Bayes nonparametric copula model. In Section \ref{sec:multi_dependence}, we derive the closed-form expression of estimates of popular multivariate dependence measures based on the novel methodology of multivariate copula estimation. We then illustrate the performance of the proposed methodology in Section \ref{sec:numerical}. \ref{sec:finite} shows the finite-sample performance for bivariate cases. The accuracy of the estimation of multivariate dependence measures is investigated in Section \ref{sec:accuracy}. Section \ref{sec:degree} illustrates the estimation of tuning parameters of the proposed ECBC copula estimator and the comparison with the empirical Bernstein copulas is provided in Section \ref{sec:comparison}.  Section \ref{sec:portfolio} provides an application to portfolio risk management. Finally, we make some general comments in Section \ref{sec:conc}.

\section{An Empirical Bayes Nonparametric Copula Model \label{sec:meth}}

Suppose we have i.i.d. samples $(X_{i1},\ldots, X_{id}) \sim F(x_1, \ldots, x_d), i \in \{1,\ldots,n\}$, where $F$ is a cumulative distribution function and $F_j$ is the absolutely continuous marginal CDF of the $j$-th component. By Sklar 's theorem (\citet{sklar1959fonctions}), there exists a unique copula $C(\cdot)$ such that
\begin{equation*}
F(x_1,\ldots, x_d) = C(F_1(x_1),\ldots, F_d(x_d)), \;\;\forall\;(x_1,\ldots,x_d) \in \mathbb{R}^d,
\end{equation*}
and 
\begin{equation*}
(F_1(X_1),\ldots, F_d(X_d)) \sim C.
\end{equation*}
The Bernstein copula is a family of copulas defined in terms of Bernstein polynomials and it was first introduced by \citet{sancetta2004bernstein}. It is a flexible model that can be used to uniformly approximate any copula. The Bernstein polynomial with degrees $(m_1, \ldots, m_d)$ of a function $C: [0,1]^d \rightarrow \mathbb{R}$ is defined as 
\begin{equation}
\label{eq:bp}
B_m(C)(\mathbf{u}) = \sum_{k_1 = 0}^{m_1} \dots \sum_{k_d = 0}^{m_d} C\left(\frac{k_1}{m_1},\ldots,\frac{k_d}{m_d}\right) \prod_{j=1}^d {m_j \choose k_j } u_j^{k_j} (1 - u_j)^{m_j - k_j},
\end{equation}
and $B_m(C)$ is called the Bernstein copula when C is a copula. 

A general estimation for the Bernstein copula of an unknown copula C is the empirical Bernstein copula (\citet{sancetta2004bernstein}) $B_m(C_n)$, where $C_n$ is the rank-based empirical copula. We denote the empirical Bernstein copula as 
\begin{equation*}
C_{m,n}(\mathbf{u}) = \sum_{k_1 = 0}^{m_1} \dots \sum_{k_d = 0}^{m_d} \hat{\theta}_{k_1,\ldots,k_d} \prod_{j=1}^d {m_j \choose k_j } u_j^{k_j} (1 - u_j)^{m_j - k_j},
\end{equation*}
where 
\begin{equation*}
\hat{\theta}_{k_1,\ldots,k_d} = C_n\left(\frac{k_1}{m_1}, \ldots, \frac{k_d}{m_d}\right) = \frac{1}{n} \sum_{i=1}^{n} \prod_{j=1}^d \mathbb{I}\left(F_{nj}^{*}(X_{ij}) \leq \frac{k_j}{m_j}\right).
\end{equation*}
where $\mathbb{I}(\cdot)$ denotes indicator function and slightly modified empirical marginal distribution functions are defined as 
\begin{equation*}
F_{nj}^{*}(x_j) =  \frac{1}{n+1} \sum_{i=1}^{n} \mathbb{I}(X_{ij} \leq x_j), \; \; \mbox{for}\;\;j \in \{1, \dots, d \},
\end{equation*}
where the modification $1/(n + 1)$ instead
of $1/n$ modifies the standard empirical marginal distribution to be away from 1 in order to reduce potential problems at boundaries.

However, the empirical Bernstein copula $C_{m,n}$ is not guaranteed to be a valid copula for finite samples as the empirical copula $C_n$ is not necessarily a genuine copula. \citet{segers2017empirical} show that the empirical Bernstein copula $C_{m,n}$ is a copula if and only if all the degrees $m_1, . . . ,m_d$
are divisors of $n$. In order to obtain a valid copula estimation for any degrees, we replace the empirical copula $C_n$ with the empirical checkerboard copula $C^{\#}_n$, which is a simple multilinear extension of the empirical copula defined as 
\begin{equation*}
C^{\#}_n (\mathbf{u}) = \frac{1}{n} \sum_{i=1}^{n} \prod_{j=1}^d \min(\max((nu_j - R_{i,j}^{(n)} + 1), 0), 1),
\end{equation*}
where $R_{i,j}^{(n)}$ is the rank of $X_{ij}$ among $X_{1j}, \ldots, X_{nj}$, see, e.g., \citet{carley2002new} and \citet{li1997approximation} for more details. Notice that the main difference between the empirical copula $C_{n}$ and the empirical checkerboard
copula $C^{\#}_n$ is that $C^{\#}_n$ is a genuine copula, so we can obtain a valid copula estimation $C^{\#}_{m,n}$ taking the form 
\begin{equation}
\label{eq:checkboard1}
C^{\#}_{m,n}(\mathbf{u}) = \sum_{k_1 = 0}^{m_1} \dots \sum_{k_d = 0}^{m_d} \tilde{\theta}_{k_1,\ldots,k_d} \prod_{j=1}^d {m_j \choose k_j } u_j^{k_j} (1 - u_j)^{m_j - k_j},
\end{equation}
where 
\begin{equation}
\label{eq:checkboard2}
\tilde{\theta}_{k_1,\ldots,k_d} = C^{\#}_n\left(\frac{k_1}{m_1}, \ldots, \frac{k_d}{m_d}\right) = \frac{1}{n} \sum_{i=1}^{n} \prod_{j=1}^d \min(\max\left(\left(n \frac{k_j}{m_j} - R_{i,j}^{(n)} + 1\right), 0), 1\right),
\end{equation}
and we call the proposed {\em empirical checkerboard Bernstein copula (ECBC)}.

Unlike the empirical Bernstein copula, the ECBC is a genuine copula for any degrees $m_1, m_2, \ldots, m_d \in \mathbb{Z}^{+}$ and any fixed sample size $n$. It is known that Bernstein
polynomials with smaller values of degrees $m_j$'s may lead to biased estimates while unnecessary larger degrees of Bernstein polynomials will necessarily lead to larger variances. So it is critical to choose the proper degrees of the ECBC based on a given sample. In order to do that, we develop an empirical Bayes method for choosing `optimal' degrees $(m_1, m_2, \ldots, m_d)$, where $m_j$'s are allowed to be different for different $j \in \{1, \dots, d \}$ and also depend on the random sample of observations.

As illustrated in \citet{sancetta2004bernstein}, using partial derivatives of (\ref{eq:checkboard1}) with respect to each $u_j$ and rearranging, we can obtain the density corresponding to ECBC as follows:
\begin{equation}
\begin{aligned}
c^{\#}_{m,n}(\mathbf{u}) &= \sum_{k_1 = 0}^{m_1 -1 } \dots \sum_{k_d = 0}^{m_d - 1} \tilde{w}_{k_1,\ldots, k_d}  \prod_{j=1}^d m_j {m_j-1 \choose k_j} u_j^{k_j} (1 - u_j)^{m_j - k_j - 1}\\
&= \sum_{k_1 = 0}^{m_1-1} \dots \sum_{k_d = 0}^{m_d-1} \tilde{w}_{k_1,\ldots, k_d}  \prod_{j=1}^d  Beta(u_j, k_j + 1, m_j - k_j),
\label{eq:density1}
\end{aligned}
\end{equation}
where 
\begin{equation}
\begin{aligned}
 \tilde{w}_{k_1,\ldots, k_d}  &= \sum_{l_1 = 0}^{1} \dots \sum_{l_d = 0}^{1} (-1)^{d+l_1+ \ldots + l_d} C^{\#}_n\left(\frac{k_1+l_1}{m_1}, \ldots, \frac{k_d+l_d}{m_d}\right). \\
\label{eq:density2}
\end{aligned}
\end{equation}
Clearly, the Bernstein copula is a mixture of independent Beta distributions leading to a tensor product form. For notational convenience, let us denote
\begin{equation*}
U_{ij} \equiv F_{nj}^{*}(X_{ij}) \; \;i \in \{1,\ldots,n\}, j \in \{1, \dots, d \},
\end{equation*}
Following the work by \citet{gijbels2010positive}, the pseudo-observations $(U_{i1}, \ldots U_{id}), i \in \{1,\ldots,n\}$ can be treated as samples from $(F_1(U_{i1}),\ldots, F_d(U_{id})) \sim C$. We then use this approximation to build an empirical Bayesian hierarchical model:
\begin{equation*}
U_{ij} | L_{ij} \stackrel{ind}{\sim} Beta(L_{ij} + 1, m_j - L_{ij}) , \; \; i \in \{1,\ldots,n\}, j \in \{1, \dots, d \},
\end{equation*}
\begin{equation*}
L_{ij} = \lfloor m_j V_{ij}\rfloor,
\end{equation*}
where $\lfloor a\rfloor$ is the `floor' function denoting the largest integer not exceeding the value $a$ and
\begin{equation*}
(V_{i1} ,\ldots,V_{id}) \stackrel{i.i.d}{\sim} C^{\#}_n(\cdot),  \; \; i \in \{1,\ldots,n\},
\end{equation*}
, i.e., $(V_{i1} ,\ldots,V_{id}), i \in \{1,\ldots,n\}$ are samples from the empirical checkerboard copula $C^{\#}_n$. It then follows that 
\begin{equation*}
\begin{aligned}
\tilde{w}_{k_1,\ldots, k_d} &= \Pr(L_{i1} = k_1, \ldots,  L_{id} = k_d) \\
&=  \Pr\left(\frac{k_1}{m_1} \leq V_{i1} < \frac{k_1+1}{m_1}, \ldots,  \frac{k_d}{m_d} \leq V_{id} < \frac{k_d+1}{m_d}\right).\\
\end{aligned}
\end{equation*}
Based on the proposition 1 in \citet{genest2017asymptotic}, $V_{ij}$ can be drawn using the following hierarchical scheme:
\begin{equation*}
\begin{aligned}
\pi_i \stackrel{i.i.d}{\sim} DisUnif\{1,\ldots,n\}, \; \; i \in \{1,\ldots,n\},
\\
\Lambda_{ij} \stackrel{i.i.d}{\sim} Unif(0, 1) \; \; i \in \{1,\ldots,n\},j \in \{1, \dots, d \}.
\\
V_{ij} = (1-\Lambda_{ij})F_{n,j}(X_{\pi_i j} -) + \Lambda_{ij}F_{n,j}(X_{\pi_i j}),
\end{aligned}
\end{equation*}
where $F_{n,j} (x_j) = 1/n \sum_{i=1}^{n} I(X_{ij} \leq x_j)$ and $DisUnif$ denotes the discrete uniform distribution, i.e., $\Pr[\pi_i=j]=1/ n$ for $j \in \{1, \dots, d \}$. Assuming that there are no ties in the pseudo samples $U_{1j}, \ldots, U_{nj}$ (owing to absolute continuity of marginal distributions or breaking it by random assignment in practice), we can equivalently represent the $V_{ij}$'s more conveniently as
\begin{equation*}
\begin{aligned}
V_{ij} &= (1-\Lambda_{ij})\frac{R_{\pi_i,j}^{(n)}-1}{n} + \Lambda_{ij}\frac{R_{\pi_i,j}^{(n)}}{n}\\
&= \frac{R_{\pi_i,j}^{(n)}-1 + \Lambda_{ij}}{n}.
\end{aligned}
\end{equation*}

Next, to account for the uncertainty in the estimation of the degrees $m_j$'s, we propose to introduce a sample-size dependent empirical prior distribution on the degrees $m_1,\ldots, m_d$ and obtain posterior estimates by Markov Chain Monte Carlo (MCMC) methods. This would not only allow for the almost automatic adaptive estimation of the degrees (based on the observed data) but also allow for quantifying the uncertainty of this crucial tuning parameter vector. Notice that the idea of putting priors on the polynomial degrees has also been adopted by \citet{burda2014copula}. However, their priors don't rely on data or sample size and they use multivariate Bernstein density instead of Bernstein copula density, i.e., the weights are belonged to a simplex without any more constraints. A Dirichlet process with a baseline of uniform distribution on $[0,1]^d$ is assigned as the prior for the copula $C$ in (\ref{eq:bp}), which doesn't guarantee $C$ to be a valid copula. In order to avoid the construction of priors under constraints, we use the empirical estimates for the coefficients of the Bernstein copula instead of assigning priors to them.

Motivated by the asymptotic theory of the empirical Bernstein estimator, e.g., as in \citet{janssen2014note}, we propose the hierarchical shifted Poisson distributions as the prior distribution for $m_1, \ldots m_d$:
\begin{equation}
\label{prior1}
m_j \mid\alpha_j\stackrel{ind}{\sim} Poisson(n^{\alpha_j}) + 1,\; \; j \in \{1, \dots, d \}.
\end{equation}
and 
\begin{equation}
\label{prior2}
\alpha_j \stackrel{i.i.d}{\sim} Unif\Big(\frac{1}{3}, \frac{2}{3}\Big), \; \;j \in \{1, \dots, d \}.
\end{equation}
The following theorem provides the large-sample consistency of the ECBC using the same set of assumptions as required for the large-sample consistency of the empirical checkerboard copula.

\begin{theorem}
\label{consist}
Given the empirical priors distribution of $m_1, \ldots,m_d$ as in (\ref{prior1}) and (\ref{prior2}), and assuming the regularity conditions for the consistency of the empirical checkerboard copula, the proposed ECBC is consistent in the following sense:
\begin{equation*}
E (|| C^{\#}_{m,n} - C|| ) = E \left[ \underset{\mathbf{u} \in [0,1]^d}{\sup} \left| C^{\#}_{m,n}(\mathbf{u}) - C (\mathbf{u}) \right| \right] \stackrel{a.s.}{\rightarrow} 0  \; \; \textit{as} \; \; n \rightarrow \infty.
\end{equation*}
where the expectation is taken with respect to the empirical prior distribution.
\end{theorem}
\begin{proof}[\textbf{\upshape Proof:}]
We denote the ECBC as $B_m(C_n^{\#})$ for simplicity. Also, the empirical 
Bernstein copula and the Bernstein copula are denoted as $B_m(C_n)$ and $B_m(C)$, respectively. Let $\displaystyle ||g||=\sup_{\mathbf{u}\in [0, 1]^d}g(\mathbf{u})$ denote the supremum norm of a function $g(\cdot)$ defined on d-dimensional square $[0, 1]^d$. Using the familiar triangle inequality we have
\begin{equation*}
|| B_m(C_n^{\#}) - C|| \leq || B_m(C_n^{\#}) -B_m(C_n)|| + || B_m(C_n) -B_m(C)|| +  ||B_m(C) - C||
\end{equation*}
First, under the assumption that the marginal CDFs are continuous, it follows from the Remark 2 in \citet{genest2017asymptotic} that
\begin{equation*}
|| C_n^{\#} - C_n|| \leq \frac{d}{n}.
\end{equation*}
Next, notice that
\begin{equation*}
\begin{aligned}
& || B_m(C_n^{\#}) - B_m(C_n)|| \\
&=  \underset{\mathbf{u} \in [0,1]^d}{\sup} \left|\sum_{k_1 = 0}^{m_1} \dots \sum_{k_d = 0}^{m_d} \left(C^{\#}_n\left(\frac{k_1}{m_1}, \ldots, \frac{k_d}{m_d}\right) - C_n\left(\frac{k_1}{m_1}, \ldots, \frac{k_d}{m_d}\right)\right) 
 \prod_{j=1}^d {m_j \choose k_j } u_j^{k_j} (1 - u_j)^{m_j - k_j} \right|\\
& \leq  \underset{0 \leq k_1 \leq m_1, \ldots 0 \leq k_d \leq m_d}{\max} \left|C^{\#}_n\left(\frac{k_1}{m_1}, \ldots, \frac{k_d}{m_d}\right) - C_n\left(\frac{k_1}{m_1}, \ldots, \frac{k_d}{m_d}\right)\right|\\
& \leq || C_n^{\#} - C_n||\leq \frac{d}{n}
\end{aligned}
\end{equation*}
In above the second inequality follows from the fact that since ${m_j \choose k_j } u_j^{k_j} (1 - u_j)^{m_j - k_j}, k_j \in \{ 0,\ldots,m_j \}$ are binomial probabilities, $\sum_{k_j=0}^{m_j}{m_j \choose k_j } u_j^{k_j} (1 - u_j)^{m_j - k_j}=1$ for any $u_j\in [0, 1]$ and for any $j \in \{1, \dots, d \}$.

\noindent Next  by using the  Lemma 1 in \citet{janssen2012large} and equation (3) in \citet{kiriliouk2019some}, we obtain
\begin{equation*}
\begin{aligned}
|| C_n - C|| & \leq ||C_n - F_n (F^{-1}_{n1}(u_1), \ldots,F^{-1}_{nd}(u_d) ) || \\+
& ||F_n (F^{-1}_{n1}(u_1), \ldots,F^{-1}_{nd}(u_d) ) - F(F^{-1}_{1}(u_1), \ldots,F^{-1}_{d} d(u_d) || \\& \leq \frac{d}{n} + O(n^{-1/2}(\ln \ln n)^{1/2})\;\;\;a.s.\\
&= O(n^{-1/2}(\ln \ln n)^{1/2})\;\;\; a.s..
\end{aligned}
\end{equation*}
Hence, it now follows that
\begin{equation*}
|| B_m(C_n) -B_m(C)|| \leq || C_n - C|| \leq O(n^{-1/2}(\ln \ln n)^{1/2}),\;\;\; a.s.
\end{equation*}
Also, by using Lemma 3.2 in \citet{segers2017empirical} we have
\begin{equation*}
|| B_m(C) - C|| \leq \sum_{j=1}^d  \frac{1}{2 \sqrt{m_j}}
\end{equation*}

\noindent Thus, combining the above inequalities that are satisfied almost surely (a.s.) for every fixed $m_j$'s we obtain
\begin{equation*}
\begin{aligned}
|| B_m(C_n^{\#}) - C|| & \leq || B_m(C_n^{\#}) -B_m(C_n)|| + || B_m(C_n) -B_m(C)|| +  ||B_m(C) - C||\\
& \leq \frac{d}{n}+ \sum_{j=1}^d  \frac{1}{2 \sqrt{m_j}} + O(n^{-1/2}(\ln \ln n)^{1/2})\;\;\; a.s.\\
\end{aligned}
\end{equation*}
Next we consider the proposed empirical priors on the degrees $m_1, \ldots, m_d$ to be
\begin{equation*}
m_j |\alpha_j\stackrel{ind}{\sim} Poisson(n^{\alpha_j}) + 1\;\;\mbox{and}\;\;
\alpha_j \stackrel{iid}{\sim} Unif\Big(\frac{1}{3}, \frac{2}{3}\Big) \; \;\mbox{for}\; j \in \{1, \dots, d \}.
\end{equation*}
We make use of the following simple Lemma:
\begin{lemma}
Suppose $M\sim Poisson(\lambda)$, then $E[{1\over\sqrt{M+1}}]\leq \sqrt{{1-e^{-\lambda}\over\lambda}}$
\end{lemma}
\noindent {\em Proof of Lemma 1:} By Jensen's inequality for the square-root function, it follows that

\[E\left({1\over\sqrt{M+1}}\right)\leq \sqrt{E\left({1\over M+1}\right)}=\sqrt{{1\over\lambda} \sum_{m=0}^\infty  {\lambda^{m + 1} e^{-\lambda}\over (m + 1)!}}  = \sqrt{{1-e^{-\lambda} \over \lambda}}\].

\noindent Notice that as $\alpha_j\sim Unif(1/ 3, 2/ 3)$, $\Pr(\alpha_j>1/ 3)=1$ and conditioning on $\alpha_j$, we then have by the above Lemma, $E\left(\sqrt{1/m_j}\big|\alpha_j\right)\leq \sqrt{(1-e^{-n^{\alpha_j}})/ n^{\alpha_j}}\rightarrow 0$ as $n \rightarrow \infty$. Thus, taking expectation with respect to the prior distribution, it follows that 
\begin{equation*}
\begin{aligned}
E  || B_m(C_n^{\#}) - C||  & \leq E  || B_m(C_n^{\#}) -B_m(C_n)||  + E  || B_m(C_n) -B_m(C)||  +  E ||B_m(C) - C|| \\
& \leq \frac{d}{n}+ \sum_{j=1}^d \frac{1}{2} E\left(\frac{1}{ \sqrt{m_j}}\right) + O(n^{-1/2}(\ln \ln n)^{1/2}) \rightarrow 0 \; as \; n \rightarrow \infty\;\; a.s.
\end{aligned}
\end{equation*}
\end{proof}

Notice that in the above result the a.s. convergence is with respect to the empirical {\em marginal} distribution of the data integrating out the conditional empirical distribution of data (given the $m_j$'s) weighted by the empirical prior distribution of the tuning parameters $m_j$'s. This is not the usual notion of posterior consistency but rather the notion can be viewed as using integrated likelihood approach (\citet{berger1999integrated}) with respect to the empirical marginal distribution obtained by integrating the priors given by equations (\ref{prior1}) and (\ref{prior2}).

It is to be noted that the joint posterior distribution of ($m_1, \ldots m_d$) may not preserve necessarily an exchangeable structure as the above prior. Using the empirical Bayes hierarchical structure of the above-proposed model it can be shown that efficient MCMC methods can be utilized to draw approximate samples from the path of a geometrically ergodic Markov Chain with posterior distribution as its stationary distribution. By generating a sufficiently large number of MCMC samples we can estimate the marginal posterior mode of the discrete-valued parameter $m_j$'s as final estimates. Let $m_{j1}, \ldots, m_{jK}$ denote $K$ MCMC samples of $m_j, j \in \{1, \dots, d \}$ and for each $j$, let $\tilde{m}_{j1}<\tilde{m}_{j2}<\cdots<\tilde{m}_{jD_j}$ denote the distinct values among these MCMC samples. Then the (marginal) posterior mode of $m_j$ is estimated by 
\begin{equation*}
\hat{m_j} =  \underset{m_{jb}, b=1,\ldots, D_j}{\text{argmax}} \sum_{a=1}^K \mathbb{I}(m_{ja} = \tilde{m}_{jb}), \; \; j \in \{1, \dots, d \}.
\end{equation*}
The final estimate of the smooth copula based on the proposed ECBC is then given by
\begin{equation*}
C^{\#}_{\hat{m},n}(\mathbf{u}) = \sum_{k_1 = 0}^{\hat{m}_1} \dots \sum_{k_d = 0}^{\hat{m}_d} \tilde{\theta}_{k_1,\ldots,k_d} \prod_{j=1}^d {\hat{m}_j \choose k_j } u_j^{k_j} (1 - u_j)^{\hat{m}_j - k_j}
\end{equation*}
where 
\begin{equation*}
\tilde{\theta}_{k_1,\ldots,k_d} = C^{\#}_n\left(\frac{k_1}{\hat{m}_1}, \ldots, \frac{k_d}{\hat{m}_d}\right). 
\end{equation*}
It is to be noted that other posterior estimates (e.g., posterior mean when it exists or coordinate-wise posterior median or some version of multivariate posterior median) can also be used but for simplicity (and the requirement that these posterior estimates of $m_j$'s be necessarily integer-valued) we have chosen to use posterior mode based on the marginal posterior distributions of $m_j$'s. Through many numerical illustrations we show the easy applicability of this choice in various examples.

\subsection{Multivariate Dependence Estimation}
\label{sec:multi_dependence}
In higher dimensions, it is often of interest to evaluate the strength of dependence among variables. This is often done using copulas since most dependence measures can be expressed as a function of copulas.  Spearman's rank correlation coefficient (Spearman's rho) is one of the most widely used dependence measures.  For a bivariate copula $C$, Spearman's rho can be written as 
\begin{equation*}
\rho = 12 \int_{0}^{1} \int_{0}^{1} C(u,v) du dv - 3= 12 \int_{0}^{1} \int_{0}^{1} (C(u,v) - uv) du dv.
\end{equation*}
A multivariate extension of Spearman's rho given in \citet{nelsen1996nonparametric} takes the form
\begin{equation}
\rho_d = \frac{\int_{I^d}  C(\mathbf{u}) d\mathbf{u} - \int_{I^d} \Pi(\mathbf{u}) d\mathbf{u}}{\int_{I^d}  M(\mathbf{u}) d\mathbf{u} - \int_{I^d}  \Pi(\mathbf{u}) d\mathbf{u}} = \frac{d+1}{2^d - (d+1)}\Big(2^d \int_{I^d}  C(\mathbf{u}) d\mathbf{u} - 1  \Big).
\label{eq:spearman_rho}
\end{equation}

Compared to vine copulas that rely on pair copulas and complex tree structures, one of the advantages of our copula estimator is that it is straightforward to obtain an estimate of multivariate Spearman's rho as 
\begin{equation}
\label{eq:rho}
\hat{\rho}_d = \frac{d+1}{2^d - (d+1)}\left(2^d \sum_{k_1 = 0}^{\hat{m}_1} \dots \sum_{k_d = 0}^{\hat{m}_d} \tilde{\theta}_{k_1,\ldots,k_d} \prod_{j=1}^d {\hat{m}_j \choose k_j } B(k_j+1, \hat{m}_j-k_j+1) - 1  \right).
\end{equation}
where $B$ is the beta function.

It can be shown that the multivariate Spearman's rho is bounded by
\begin{equation*}
\frac{2^d - (d+1)!}{d! (2^d - d - 1)} \leq \rho_d \leq 1,
\end{equation*}
where the lower bound approaches to zero as dimension increases. Since our copula estimator is a genuine copula, the estimate of multivariate d-dimensional Spearman's rho $\hat{\rho}_d$ can avoid taking values out of the parameter space, which might be an issue for estimates built on other nonparametric copula estimators, e.g., the empirical copula (see \citet{perez2016note}). 

Similar to Spearman's rho, Kendall's tau is another common dependence measure and has its multivariate version as well, which is given by \citet{nelsen1996nonparametric} as
\begin{equation}
\tau_d = \frac{1}{2^{d-1} - 1} \left(2^d \int_{I^d}  C(\mathbf{u}) d C(\mathbf{u}) - 1  \right).
\label{eq:kendall_tau}
\end{equation}
By applying (\ref{eq:density1}) and (\ref{eq:density2}), it is also easy to obtain an estimate of multivariate Kendall's tau based on our copula estimator as 
\begin{equation*}
\begin{aligned}
\hat{\tau}_d = \frac{1}{2^{d-1} - 1}\Big(2^d \sum_{k_1 = 0}^{\hat{m}_1 -1} \dots \sum_{k_d = 0}^{\hat{m}_d -1} \sum_{l_1 = 0}^{\hat{m}_1} \dots \sum_{l_d = 0}^{\hat{m}_d} \tilde{w}_{k_1,\ldots,k_d} \tilde{\theta}_{l_1,\ldots,l_d} \\
\prod_{j=1}^d \hat{m}_{j}{\hat{m}_j -1 \choose k_j } {\hat{m}_j \choose l_j } B(k_j + l_j+1, 2\hat{m}_j-k_j-l_j) - 1  \Big).
\end{aligned}
\end{equation*}

Thus, using our proposed ECBC copula not only are we able to obtain a fully non-parametric estimate of any copula function in closed form (once the tuning parameters $m_j, j \in \{1, \dots, d \}$ are estimated by their posterior modes) but also we are able to derive the closed-form expression of estimates of the popular multivariate measures of dependence for any arbitrary dimension $d\geq 2$. 

Moreover, although we only illustrate the use of multivariate extensions of Kendall's tau and Spearman's rho as possible measures of multivariate dependence, any other multivariate notion of dependence measures that are suitable functionals of the underlying copula can also be computed using our closed-form expression of the ECBC estimator. This is particularly advantageous compared to even some of the flexible yet complicated parametric copula family (e.g., Archimedian, multivariate Gaussian or t, etc.)  for which it may require high-dimensional numerical integration to compute multivariate versions of Spearman's rho as given in (\ref{eq:spearman_rho}) and/or Kendall's tau given in (\ref{eq:kendall_tau}). For vine copulas, it is particularly challenging to obtain estimates of these multivariate measures of dependence as such high-dimensional integrals are often algebraically and even numerically intractable say for dimension $d\geq 5$, whereas for ECBC even when a new measure of dependence is created as a functional of the copula that may be more complicated than those defined in eq.s (\ref{eq:spearman_rho}) and (\ref{eq:kendall_tau}), we can easily obtain a large number of Monte Carlo (MC) samples from the ECBC and use MC based approximation to estimate such new measures of multivariate dependence (we illustrate such a case in our real case study involving portfolio risk optimization in Section \ref{sec:portfolio}).

\section{Numerical Illustrations using Simulated Data}
\label{sec:numerical}

\subsection{Finite-Sample Performance for Bivariate Cases}
\label{sec:finite}

We investigate the finite-sample performance of the ECBC through a Monte Carlo simulation study. Samples from the true copula are generated using the package $\tt copula$ in $\tt R$ (\citet{hofert2014package}). In order to visualize the results using contour plots, we first restrict our illustration to bivariate copulas. Three copula families with various parameters and an asymmetric copula are considered. 
The first four examples are the bivariate Frank copulas
\begin{equation*}
C_F(u,v) = - \frac{1}{\theta} \ln \Big(1 + \frac{(\text{exp}(- \theta u) - 1) (\text{exp}(- \theta v) - 1)}{\text{exp}(- \theta) - 1)} \Big),
\end{equation*}
with parameter $\theta$ equal to $-2, -1, 1$ and $2$, which reflects a wide range of dependence from negative to positive. 
The next two examples are the Clayton copula with parameter $1$
\begin{equation*}
C_C(u,v) = (\max\{ u^{-1}+ v^{-1} - 1, 0 \})^{- 1},
\end{equation*}
and the Gumbel copula with parameter $2$
\begin{equation*}
C_G(u,v) = \text{exp}(-((-\ln(u))^{2} + (-\ln(v))^{2})^{1/ 2}).
\end{equation*}
The value of Kendall's tau is 0.33 for Clayton copula with parameter $1$ and 0.5 for Gumbel copula with parameter $2$. Both cases have a moderate positive dependence. 

The next example is the independent copula $C(u, v) = uv$. Finally, we consider an asymmetric copula 
\begin{equation*}
C_a(u,v) = u v - 0.12 (1 - v^2) \text{sin}(8.3 v) u (1-u),
\end{equation*}

In the simulation study,  $n = 100$ samples are drawn from the true copula for each replicate (of size $n = 100$) and there are $N=100$ replicates. Degrees of the ECBC are estimated by
posterior modes by obtaining 5000 MCMC samples following 2000 burn-in samples of two chains for each of the $N$ replicated data sets generated from a chosen true copula model. It is to be noted that it takes about 2, 11, and 60 minutes to run 7000 iterations for two chains by using MacOS with 16GB of RAM for $d=2$, $d=10$, and $d=50$, respectively, when data are drawn from multivariate Frank copula. Convergence of MCMC runs was monitored based on preliminary runs using standard diagnostics available in \textsf{R} packages \texttt{rjags} and \texttt{coda}. We show the results for eight copulas in Fig. \ref{fig:estimation1} and \ref{fig:estimation2}. The contour plot of the true underlying copula and the empirical MC average of $N$ copula estimates are given for comparison and $(\bar{m}_1, \bar{m}_2)$ represents the MC mean of the posterior modes of degree parameters.

We can see from the contour plots that the average of the estimated copula is extremely close to the underlying true copula across all different dependence structures irrespective of the assumed parametric models. This illustrates that the proposed ECBC has a robust performance in estimating various true copulas.

\begin{figure}
    \centering
    \begin{subfigure}[b]{0.35\textwidth}
        \includegraphics[width=\textwidth]{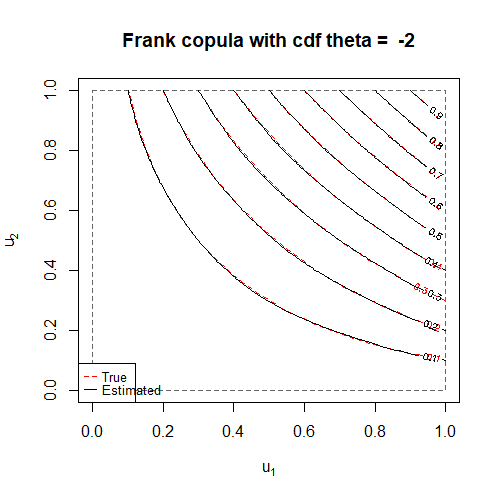}
        \caption{$(\bar{m}_1, \bar{m}_2)=(22.52, 24.01)$}
    \end{subfigure}
    \begin{subfigure}[b]{0.35\textwidth}
        \includegraphics[width=\textwidth]{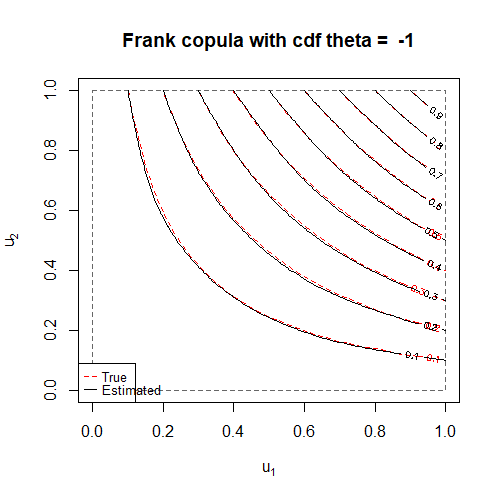}
        \caption{$(\bar{m}_1, \bar{m}_2)=(22.46, 23.17)$}
    \end{subfigure}
    \begin{subfigure}[b]{0.35\textwidth}
        \includegraphics[width=\textwidth]{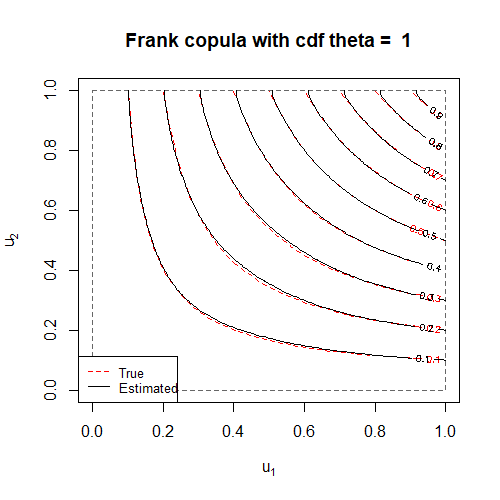}
        \caption{$(\bar{m}_1, \bar{m}_2)=(23.29, 24.83)$}
    \end{subfigure}
        \begin{subfigure}[b]{0.35\textwidth}
        \includegraphics[width=\textwidth]{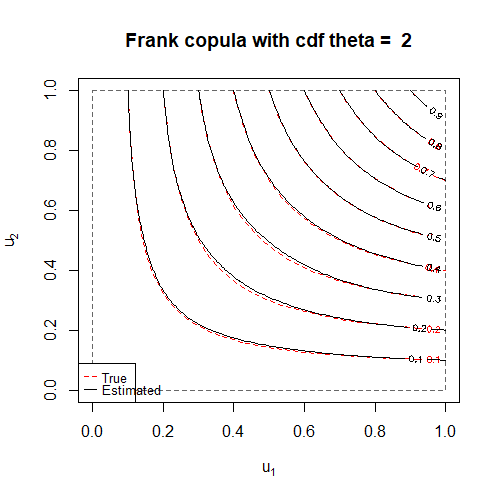}
        \caption{$(\bar{m}_1, \bar{m}_2)=(22.82, 22.17)$}
    \end{subfigure}
    \caption{Estimation of Frank copulas using the ECBC with empirical Bayesian method for choosing proper degrees when sample size $n=100$.}\label{fig:estimation1}
\end{figure}

\begin{figure}
    \centering
        \begin{subfigure}[b]{0.35\textwidth}
        \includegraphics[width=\textwidth]{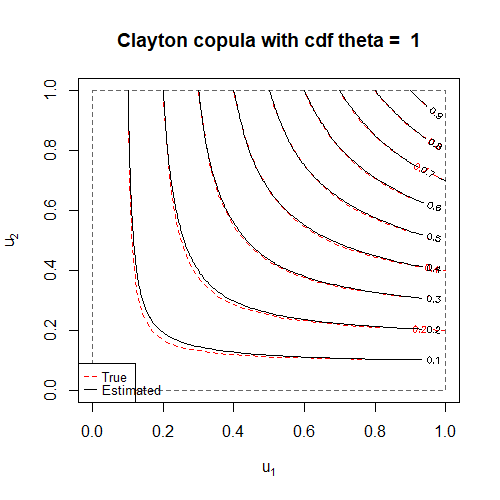}
        \caption{$(\bar{m}_1, \bar{m}_2)=(23.94, 21.69)$}
    \end{subfigure}
        \begin{subfigure}[b]{0.35\textwidth}
        \includegraphics[width=\textwidth]{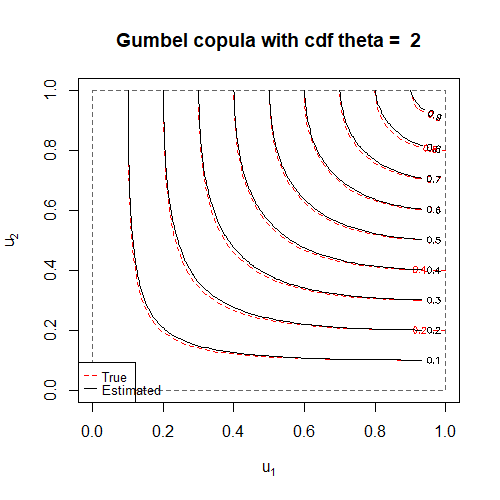}
        \caption{$(\bar{m}_1, \bar{m}_2)=(22.58, 22.16)$}
    \end{subfigure}
                \begin{subfigure}[b]{0.35\textwidth}
        \includegraphics[width=\textwidth]{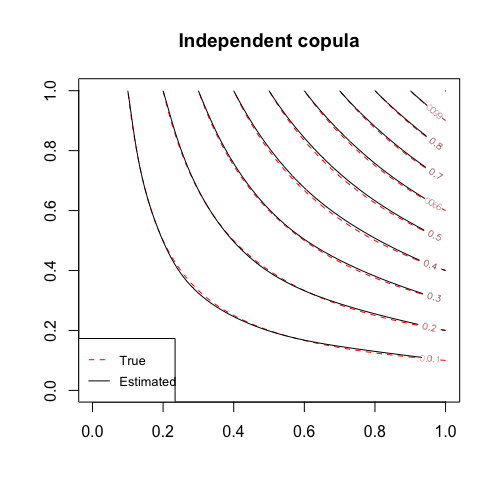}
        \caption{$(\bar{m}_1, \bar{m}_2)=(22.34, 20.87)$}
    \end{subfigure}
            \begin{subfigure}[b]{0.35\textwidth}
        \includegraphics[width=\textwidth]{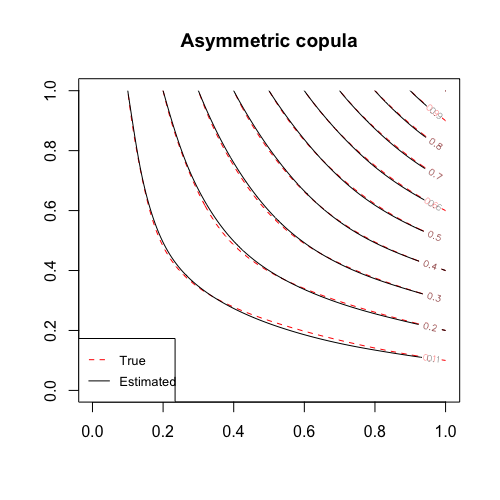}
        \caption{$(\bar{m}_1, \bar{m}_2)=(21.08, 22.92)$}
    \end{subfigure}
    
    \caption{Estimation of various copulas using the ECBC with empirical Bayesian method for choosing proper degrees when sample size $n=100$.}\label{fig:estimation2}
\end{figure}

\subsection{Accuracy of Multivariate Dependence Measures ($d=3$)}
\label{sec:accuracy}

To assess the finite sample performance of the estimate of multivariate Spearman's rho $\hat{\rho}_d$, we conduct Monte Carlo simulations for $d=3$ copulas. We consider the independent copula and Clayton copulas with parameter value $\{ 0.5, 1, 2 \}$, respectively. For each copula model, $N = 100$ Monte Carlo replicates are generated with size $n=100$. For each replicate, we compute the proposed estimator $\hat{\rho}_d$ in (\ref{eq:rho}) and the estimator based on empirical copula $\tilde{\rho}_d$ 
\begin{equation}
\begin{aligned}
\tilde{\rho}_d &= \frac{d+1}{2^d - (d+1)}\Big(2^d \int_{I^d} C_n(\mathbf{u}) d\mathbf{u} - 1  \Big) \\
&= \frac{d+1}{2^d - (d+1)}\Big( \frac{2^d}{n}\sum_{i=1}^n \prod_{j=1}^d (1 - U_{ij}) - 1\Big).
\end{aligned}
\end{equation}
Finally, for each estimator we compute the mean, bias, variance and mean square error (MSE) over all replicates.

Table \ref{tab:rho} shows the results on the estimation of multivariate Spearman's rho for four different copulas. An approximated value of the true multivariate Spearman's rho $\rho_d$ can be obtained by numerical integration since there is no analytical expression as a function of the parameter (see \citet{perez2016note}), which is also given in Table \ref{tab:rho}. The corresponding boxplots for the two estimates based on $N = 100$ replicates along with a horizontal line for true multivariate Spearman's rho $\rho_d$ are shown in Fig. \ref{fig:rho}. 

From the results we can see our estimator $\hat{\rho}_d$ outperforms $\tilde{\rho}_d$ with respect to variance and MSE. In terms of bias, $\hat{\rho}_d$ tends to underestimate and have a larger bias as strength of dependence increases, but there is not a clear superiority of one estimator over the other. As shown in Fig. \ref{fig:rho} (c) and (d) where there is a moderate or strong dependence in trivariate cases, $\tilde{\rho}_d$ can take values out of parameter space $[-2/3, 1]$ ($3\%$ and $12\%$ of $\tilde{\rho}_d$ are taking values greater than $1$ in (c) and (d), respectively), which can be problematic in measuring dependence.

\begin{table}[htbp]
\caption{Comparison of two estimates of multivariate Spearman's rho based on  $N = 100$  replications of size $n=100$ generated from  the independent copula and Clayton copulas with different parameter values when dimension $d = 3$.} 
\centering 
\begin{tabular}{ccccccc} 
\hline\hline 
Copula & $\rho_d$  & Estimator  & Mean & Bias & Variance &MSE \\ [0.5ex]
\hline \hline
Independent & 0 & $\hat{\rho}_d$ &0.007 & 0.007  &0.008  & 0.009 \\[0.5ex]
 & & $\tilde{\rho}_d$ & -0.015 &-0.015 &0.014 &0.015 \\[0.5ex]

\hline 
Clayton & 0.308 & $\hat{\rho}_d$  &0.294 & -0.014& 0.007& 0.008   \\[0.5ex]
$\theta = 0.5$ & & $\tilde{\rho}_d$  & 0.287& -0.021& 0.028& 0.029  \\[0.5ex]

\hline
Clayton &0.504 &$\hat{\rho}_d$ &0.477 & -0.027 & 0.007 & 0.008 \\[0.5ex]
$\theta = 1$ & & $\tilde{\rho}_d$ &0.519 & 0.015 & 0.039 & 0.040\\[0.5ex]

\hline
Clayton  & 0.717& $\hat{\rho}_d$ &0.680 & -0.037 & 0.004 & 0.006 \\[0.5ex]
$\theta = 2$& & $\tilde{\rho}_d$  & 0.732 &  0.015 & 0.050 & 0.051  \\[0.5ex]

\hline \hline 
\end{tabular}
\label{tab:rho}
\end{table}

\begin{figure}
    \centering
        \begin{subfigure}[b]{0.35\textwidth}
        \includegraphics[width=\textwidth]{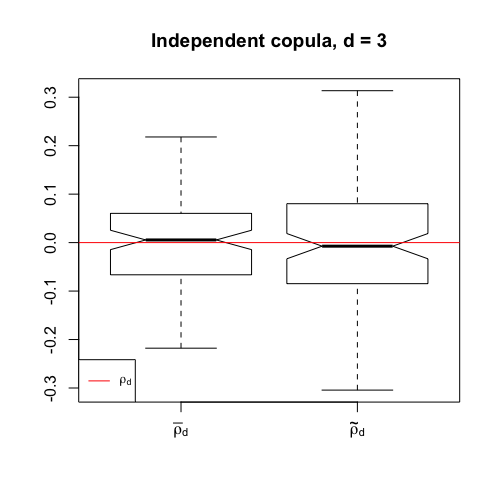}
        \caption{}
    \end{subfigure}
        \begin{subfigure}[b]{0.35\textwidth}
        \includegraphics[width=\textwidth]{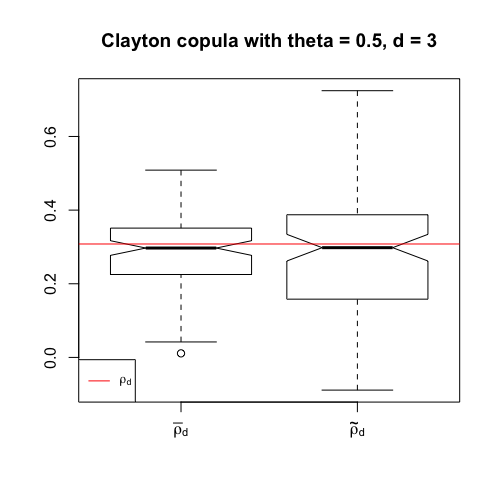}
        \caption{}
    \end{subfigure}
                \begin{subfigure}[b]{0.35\textwidth}
        \includegraphics[width=\textwidth]{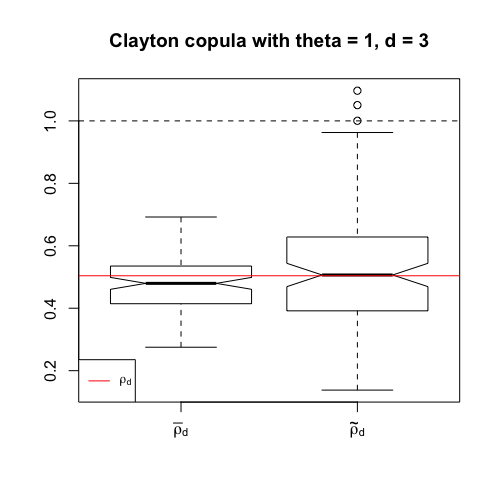}
        \caption{}
    \end{subfigure}
            \begin{subfigure}[b]{0.35\textwidth}
        \includegraphics[width=\textwidth]{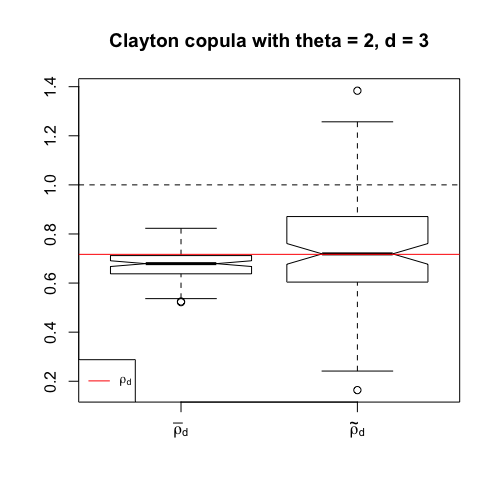}
        \caption{}
    \end{subfigure}
    
    \caption{Boxplots of the two estimators based on $N = 100$ replicates and a horizontal line of true multivariate Spearman's rho $\rho_d$ for each of four trivariate copulas.}\label{fig:rho}
\end{figure}

\subsection{Estimation of Tuning parameters of ECBC}
\label{sec:degree}

We now illustrate one of the primary advantages of our proposed empirical Bayes estimate of the ECBC that allows for data-dependent automatic selection of dimension-varying degree parameters $m_j$'s. We further explore the special case of choosing equal degrees $m_1 =  \ldots = m_d = m$ by using the following prior distribution
\begin{eqnarray}
\label{prior3}
m \mid\alpha \sim Poisson(n^{\alpha}) + 1,\\
\mbox{and}\;\;\alpha \sim Unif\Big(\frac{1}{3}, \frac{2}{3}\Big)
\label{prior4}
\end{eqnarray}
For our empirical illustration, we consider three true bivariate copulas and one trivariate copula to explore the comparative performance of choosing dimension-varying degrees compared to setting them equal across all dimensions. The first example is the Farlie-Gumbel-Morgenstern (FGM) copula $C_a(u,v) = u v(1 - a(1 - u)(1 - v))$ with parameter $a=-1$. The next two choices are the independent copula (e.g., FGM with $a=0$) and the Gaussian copula with positive dependence (correlation $\rho = 0.5$). The last one is the trivariate t-copula with degrees of freedom $4$ and pairwise dependence $\rho_{12} = -0.2, \rho_{13} = 0.5$ and $\rho_{23} = 0.4$. Again samples of size $n = 100$ are obtained for each four cases and repeated $N=100$ times for MC evaluation. For each sample, we choose the degrees of the ECBC by computing the posterior modes using our proposed empirical Bayesian method. 

Fig. \ref{fig:degree} presents the scatterplot of estimated values of $(m_1, m_2)$ or $(m_1, m_2, m_3)$ for each chosen true copula model. From the plots we can observe that for bivariate copulas, posterior estimates of $m_1$ and $m_2$ are significantly different in most cases without any prior restrictions of equality. In fact, posterior probability of choosing equal $m_1=m_2$ is only about $0.08, 0.08$ and $0.13$ for FGM copula, for independent copula, and for Gaussian copula, respectively indicating against forcing $m_1=m_2$ as is popularly done in the literature. For the trivariate t-copula case, the posterior probability of $m_1=m_2=m_3$ is 0, decisively suggesting that equality assumption is sub-optimal in general and particular as the dimension increases. We have conducted further studies with dimensions (not shown here for limitation of space) and the conclusions remain very similar. 

\begin{figure}
    \centering
        \begin{subfigure}[b]{0.35\textwidth}
        \includegraphics[width=\textwidth]{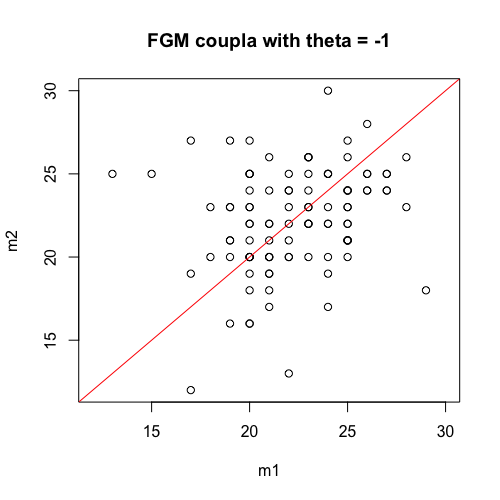}
        \caption{}
    \end{subfigure}
        \begin{subfigure}[b]{0.35\textwidth}
        \includegraphics[width=\textwidth]{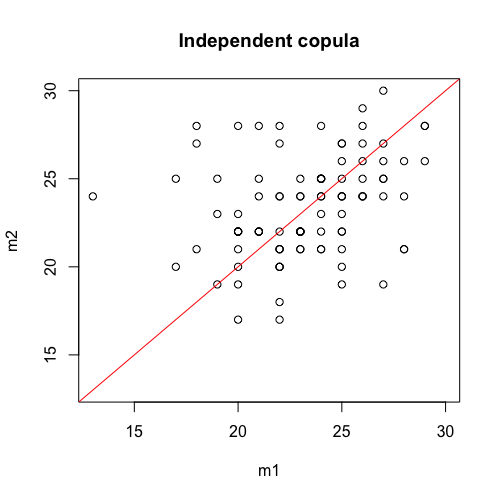}
        \caption{}
    \end{subfigure}
            \begin{subfigure}[b]{0.35\textwidth}
        \includegraphics[width=\textwidth]{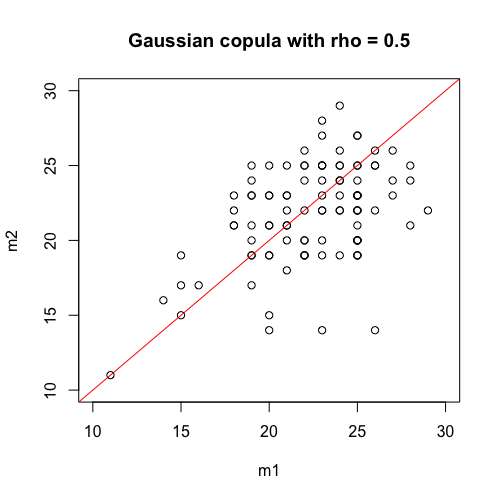}
        \caption{}
    \end{subfigure}
     \begin{subfigure}[b]{0.35\textwidth}
        \includegraphics[width=\textwidth]{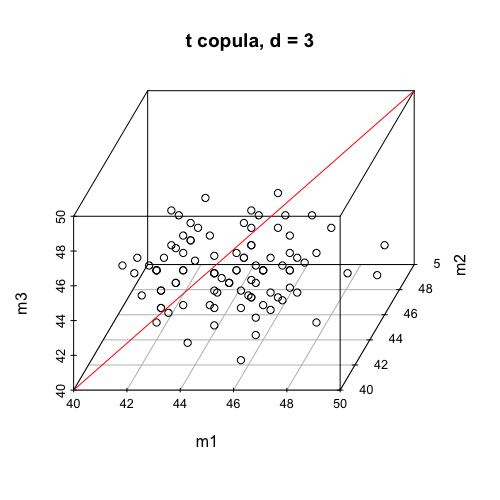}
        \caption{}
    \end{subfigure}
    \caption{Choice of dimension-varying degrees obtained by applying the proposed empirical Bayesian method based on  $N = 100$ replications when sample size $n=100$.}\label{fig:degree}
\end{figure}

In order to further compare the performances of copula estimators with flexible degrees vs. equal degrees, we fit our Bayesian models with two different settings of priors: (i) (flexible) the original priors given in (\ref{prior1}) and (\ref{prior2}) where degrees are allowed to vary in different dimensions; (ii) (equal) modified priors given in (\ref{prior3}) and (\ref{prior4}) where degrees are set to be equal; using the same data set generated from the three bivariate copulas. Following \citet{segers2017empirical}, we consider three global performance measures: the integrated squared bias, the integrated variance, and the integrated mean squared error. Given a copula estimator $\hat{C}_n$, the performance measures are defined as 
\begin{eqnarray*}
\mbox{integrated squared bias:}\;\; ISB &=& \int_{[0,1]^d} \big[E[\hat{C}_n (\mathbf{u}) - C(\mathbf{u})]\big]^2 d \mathbf{u},\\
\mbox{integrated variance:}\;\; IV &=& \int_{[0,1]^d} E\big[[\hat{C}_n (\mathbf{u}) - E(\hat{C}_n (\mathbf{u}))]^2\big]d \mathbf{u},\\
\mbox{integrated mean squared error:}\;\; IMSE &=& \int_{[0,1]^d} E\big[[\hat{C}_n (\mathbf{u}) - C(\mathbf{u})]^2\big]d \mathbf{u}.
\end{eqnarray*}
We compute the performance measures by applying the computation method described in \citet{segers2017empirical}, which relies on Monte Carlo simulation to get a Monte Carlo estimate of each performance measure. Table \ref{tab:degrees} presents the results for the four cases, where the first three are bivariate ($d=2$) and the last one is trivariate ($d=3$). The standard errors of the Monte Carlo estimates are not reported in the table as they are negligibly small. We can see that the copula estimators with flexible degrees perform better than those with equal degrees in terms of IV and IMSE in all cases. As the difference in IMSE is dominated by the IV term, the choice of flexible degrees leads to smaller uncertainty, and hence smaller IMSE while the biases remain relatively unaffected. It is also interesting to observe that estimated degrees are far smaller than sample sizes indicating the empirical beta copula may not have optimal performance, which we next explore.

\begin{table}[htbp]
\caption{Comparison of copula estimators with flexible degrees vs.  equal degrees using three performance measures computed by Monte Carlo simulation based on  $N = 100$ replications when sample size $n = 100$.} 
\centering 
\begin{tabular}{ccccc} 
\hline\hline 
Copula & Choice of degrees  & ISB ($\times 10^{-4}$) & IV ($\times 10^{-4}$) & IMSE ($\times 10^{-4}$)\\ [0.5ex]
\hline \hline
FGM & Flexible   &$0.10 $ & $1.28 $ &  $1.38 $ \\[0.5ex]
$\theta = -1$ &Equal  & $0.09$ & $1.37$ & $1.46 $ \\[0.5ex]

\hline 
Independent & Flexible   &$0.13$ & $1.89 $ &  $2.02 $\\[0.5ex]
 & Equal  & $0.20$ & $2.16 $ & $2.36 $\\[0.5ex]

\hline
Gaussian & Flexible  &$0.36$ & $0.71 $ & $1.07 $\\[0.5ex]
$\rho = 0.5$ & Equal  & $0.30$ & $0.83 $ & $1.13 $ \\[0.5ex]

\hline
t  & Flexible  &$0.11 $ & $2.67 $ & $2.78$\\[0.5ex]
$d = 3$ & Equal  &$0.19 $ & $2.79$ & $2.98 $ \\[0.5ex]

\hline \hline 
\end{tabular}
\label{tab:degrees}
\end{table}

\subsection{Comparison with the Empirical Bernstein Copulas} 
\label{sec:comparison}
In this section, we compare the finite-sample performance of the ECBC with other nonparametric copula estimators only as it has been already demonstrated that parametric models based methods lead to biased estimates under model misspecification. First, we consider the empirical beta copula introduced by \citet{segers2017empirical}, which is a special case of the empirical Bernstein copula where the degrees of the polynomials are set equal to the sample size. The empirical beta copula is a genuine copula and has been shown to outperform the classical empirical copula and the empirical checkerboard copula in terms of bias and variance. For bivariate cases, we also include the empirical Bernstein copula with degrees as suggested in \citet{janssen2012large} into the comparison. By setting degrees $m_1=m_2=m$ and minimizing the asymptotic pointwise mean squared error with respect to $m$, \citet{janssen2012large} suggested the choice of $m$ in the bivariate case as 
\begin{equation}
m_0(u_1, u_2) = \Big( \frac{4b^2(u_1, u_2)}{V(u_1, u_2)}\Big )^{2/3} n^{2/3}
\label{eq:opt_degree}
\end{equation}
where 
\begin{eqnarray*}
b^2(u_1, u_2) &=& \frac{1}{2}  \sum_{j=1}^2 u_j(1-u_j)C_{u_j u_j}(u_1, u_2)\\
\mbox{and}\;\;V(u_1, u_2) &=& \sum_{j=1}^2 C_{u_j}(u_1, u_2) (1-C_{u_j}(u_1, u_2)) \Big( \frac{u_j (1-u_j)}{\pi} \Big)^{1/2}
\end{eqnarray*}
and $C_{u_j}$ and $C_{u_j u_j}$ are the first-order and second-order partial derivatives of C, respectively, with respect to $u_j, j = 1, 2$. Note that even we use the integer part $\lfloor m_0(u_1, u_2)\rfloor$ in practice, it is not necessarily a divisor of $n$, meaning the empirical Bernstein copula with $m = \lfloor m_0(u_1, u_2)\rfloor$ is not guaranteed to be a genuine copula.

We consider the same copula models as in Section \ref{sec:degree}. The choice of degrees $m_0(u_1, u_2)$ suggested by \citet{janssen2012large} is not defined for the independent copula as $C_{u_j u_j} = 0$ and it is restricted to bivariate cases, so we only take the empirical Bernstein copula into consideration for the bivariate FGM and Gaussian copulas. All results in Table \ref{tab:comparison} are based on $N= 100$ MC replications each of sample sizes $n = 25, 50, 100$. We compare the performance of the ECBC with flexible degrees (referred to as flexible ECBC), the empirical beta copula (referred to as Beta), and the empirical Bernstein copula with $m = m_0(u_1, u_2)$ (referred to as Bernstein) using the same performance measures as in Section \ref{sec:degree}.

Table \ref{tab:comparison} indicates that the ECBC with flexible degrees outperforms the empirical beta copula in terms of variance and mean square error in all cases. Compared to the empirical Bernstein copula with $m = m_0(u_1, u_2)$, the ECBC with flexible degrees has a smaller bias but the ordering with respect to mean square error is not clear between these two copula estimators. For small samples, the empirical beta copula seems to have the largest variance while the bias of the empirical Bernstein copula with $m = m_0(u_1, u_2)$ is shown to be the largest even though it uses optimal "true" degree given in (\ref{eq:opt_degree}).

\begin{table}[htbp]
\caption{Comparison of the ECBC with flexible degrees (referred to as flexible ECBC), the empirical beta copula(referred to as Beta) and the empirical Bernstein copula with $m = m_0(u_1, u_2)$ (referred to as Bernstein) using three performance measures computed by Monte Carlo simulation based on $N= 100$ replications for sample size $n= 25, 50, 100$.} 
\centering 
\begin{tabular}{ccccccccccc} 
\hline\hline 
 &  & \multicolumn{3}{c}{ISB ($\times 10^{-4}$)} &  \multicolumn{3}{c}{IV ($\times 10^{-4}$)} &  \multicolumn{3}{c}{IMSE ($\times 10^{-4}$)}\\ [0.5ex]
Copula & Estimator & \multicolumn{3}{c}{ n =} & \multicolumn{3}{c}{n =} & \multicolumn{3}{c}{n =}\\
& & 25 & 50 & 100 & 25 & 50 & 100 & 25 & 50 & 100 \\
\hline \hline
FGM & flexible ECBC  & 0.29 & 0.07& 0.10 & 3.46& 2.72& 1.28 &3.75 & 2.79 & 1.38\\[0.5ex]
$\theta = -1$ &Beta  & 0.30 & 0.09 & 0.02 & 5.24 & 3.27 & 2.17 & 5.54&3.36 & 2.19 \\[0.5ex]
&Bernstein  & 0.83 & 0.51 & 0.07 & 1.52 & 1.33 & 1.08 & 2.35 & 1.84 &  1.15 \\[0.5ex]
\hline 
Independent & flexible ECBC  &0.68 & 0.31 & 0.13&2.37 &2.33 &1.89 & 3.05& 2.64 &2.02\\[0.5ex]
 & Beta  & 0.02 & 0.03 & 0.01 & 5.11 &4.02 & 2.25 & 5.13 & 4.05 &2.26\\[0.5ex]
&Bernstein  & NA & NA & NA & NA & NA & NA & NA & NA & NA \\[0.5ex]
\hline
Gaussian & flexible ECBC  &1.72 & 0.79 & 0.36 & 2.51& 1.67 &0.71 & 4.23& 2.46 & 1.07\\[0.5ex]
$\rho = 0.5$ & Beta  & 0.26 & 0.16& 0.20& 4.78 & 3.03 &1.81 & 5.04 & 3.19 & 2.01\\[0.5ex]
&Bernstein  & 6.82 & 3.09& 1.27 & 2.28 & 1.42& 1.06& 9.10 &4.51 & 2.33 \\[0.5ex]
\hline
t  & flexible ECBC & 0.41  & 0.30 & 0.11&  6.76&  4.26 &2.67 & 7.17& 4.56 & 2.78 \\[0.5ex]
$d = 3$ & Beta  & 0.27 & 0.15&0.04 & 7.55 & 5.05 & 3.07 & 7.82 & 5.20 & 3.11\\[0.5ex]
&Bernstein  & NA & NA & NA & NA & NA & NA & NA & NA & NA  \\[0.5ex]
\hline \hline 
\end{tabular}
\label{tab:comparison}
\end{table}

\section{Application to Portfolio Risk Management}
\label{sec:portfolio}

Copulas have been widely used in portfolio optimization and risk measurement as they are powerful tools to model the dependence among different assets in a portfolio. The proposed ECBC is capable of estimating multivariate copula and it is straightforward to sample from the estimated copula, so it can be applied to find optimal weights and estimate risk measures for a portfolio with a variety of assets.

We now illustrate the use of ECBC for portfolio risk allocation using real data consisting of a $d$ asset values. Value at risk (VaR) and conditional value at risk (CVaR) (also called expected shortfall (ES)) are common measures of risk in the field of risk management (see, e.g., \citet{jorion2007value} and \citet{uryasev2000conditional}). Assume that $X$ is the return of a portfolio or asset (daily log-return of a portfolio of stocks or individual stocks, with positive indicating profit and negative values representing loss) with distribution function $F_X(\cdot)=\Pr[X\leq x]$. The VaR of $X$ at the level of $\alpha\in (0, 1)$ is defined as
\begin{equation*}
\begin{aligned}
VaR_{\alpha}(X) = - \inf\{x \in \mathbb{R}: F_X(x) > \alpha\},
\end{aligned}
\end{equation*}
while the CVaR (ES) of $X$ is defined as 
\begin{equation*}
\begin{aligned}
CVaR_{\alpha}(X) = E(- X | X \leq - VaR_{\alpha}(X)).
\end{aligned}
\end{equation*}
Notice that, if we consider the corresponding loss of the same portfolio represented by $Y = -X$, then we have $VaR_{\alpha}(X) =VaR_{\alpha}(Y) = F_Y^{-1}(1 - \alpha)$ and $CVaR_{\alpha}(X) = CVaR_{\alpha}(Y) = E(Y | Y \geq VaR_{\alpha}(Y))$.

Mean-CVaR portfolio optimization is a popular portfolio optimization technique introduced by \citet{rockafellar2000optimization}. The advantage of Mean-CVaR portfolio optimization is that it calculates VaR and minimizes CVaR simultaneously where the optimization can be formulated as a linear programming problem.

Let $\mathbf{x} \in \mathbb{R}^d$ denote a realized return value of $d$ assets in a portfolio, and $\mathbf{v} \in {\cal S}_d=\{\mathbf{v}\in\mathbb{R}^d: v_j\geq 0,\;\forall j,\; \sum_{j=1}^{d} v_j =1\}$, denote the portfolio weights to be determined within the d-dimensional simplex ${\cal S}_d$.  The key to the approach in \citet{rockafellar2000optimization} is the auxiliary function for CVaR taking the form of  
\begin{equation}
\begin{aligned}
H_{\alpha}(\mathbf{v}, \gamma) = \gamma + \frac{1}{\alpha} \int_{l(\mathbf{v}, \mathbf{x}) \geq \gamma} (l(\mathbf{v}, \mathbf{x}) - \gamma) dF(\mathbf{x}),
\end{aligned}
\label{eq:cvar_int}
\end{equation}
where $l(\mathbf{v}, \mathbf{x}) = - \mathbf{v}^T \mathbf{x}$ is a linear loss function and $F(\mathbf{x})$ is the joint distribution function of daily (random) return vector $\mathbf{X}$ which we will estimate using our proposed ECBC based empirical Bayes method. It has been shown in Theorem 1 of \citet{rockafellar2000optimization} that for any weights $\mathbf{v}$, $H_{\alpha}(\mathbf{v}, \gamma)$ is convex as a function of $\gamma$ and is equal to $CVaR_\alpha(\mathbf{v})$ at the minimum point. Moreover, $VaR_\alpha(\mathbf{v})$ would be the left endpoint of $\underset{\gamma}{\arg\min} H_{\alpha}(\mathbf{v}, \gamma)$. Moreover, Minimizing $CVaR_\alpha(\mathbf{v})$ with respect to $\mathbf{v}$ is equivalent to minimizing $H_{\alpha}(\mathbf{v}, \gamma)$ with respect to $(\mathbf{v}, \gamma)$ (e.g., see Theorem 2 of \citet{rockafellar2000optimization} for details). To numerically approximate the integral in (\ref{eq:cvar_int}), it is often good enough to generate $M$ samples from $F(\cdot)$ or its estimate, which can be done by using the proposed empirical Bayes method based on the ECBC. However, a relatively less answered question in finance is how large should we choose $M$ for accurate estimation as the integral in (\ref{eq:cvar_int}) depends on sampling the tail part of $F(\cdot)$ or its estimate.
The empirical estimate of $H_{\alpha}(\mathbf{v}, \gamma)$ based on generating $\mathbf{x}_k\stackrel{iid}{\sim} F$ or $\hat{F}$ can be written as 
\begin{equation}
\begin{aligned}
F_{\alpha}(\mathbf{v}, \gamma) = \gamma + \frac{1}{\alpha M} \sum_{k=1}^M (- \mathbf{v}^T \mathbf{x}_k - \gamma)_+
\end{aligned}
\label{eq:cvar_app}
\end{equation}
where $(\cdot)_+ = \max(\cdot, 0)$. 
\begin{proposition}
In order to achieve an accuracy of $\epsilon > 0$ for the MC approximation, it is sufficient to generate $M$ MC samples such that 
\begin{equation}
\begin{aligned}
\sqrt{\frac{2 \ln \ln M}{M} \lambda_{max}} \leq \epsilon, \; \; i.e., \; \; \frac{M}{\ln \ln M} \geq \frac{2 \lambda_{max}}{\epsilon^2}
\end{aligned}
\label{eq:lil_bound}
\end{equation}
where $\Sigma = Var_F(\mathbf{X})$ and $\lambda_{max}$ is the largest eigenvalue of $\Sigma$. 
\end{proposition}
\begin{proof}[\textbf{\upshape Proof:}]
By the Law of the Iterated Logarithm (see, e.g., \citet{balsubramani2014sharp}), the deviation of MC approximation from the mean is almost surely bounded by 
\begin{equation*}
\begin{aligned}
\sqrt{\frac{2 \ln \ln M}{M}} \sqrt{Var((- \mathbf{v}^T \mathbf{x}_k - \gamma)_+)}
\end{aligned}
\end{equation*}
Let $\Sigma = Var_F(\mathbf{X})$ and $\lambda_{max}$ be the largest eigenvalue of $\Sigma$, then we have 
\begin{equation*}
\begin{aligned}
Var((- \mathbf{v}^T \mathbf{X} - \gamma)_+) \leq Var( \mathbf{v}^T \mathbf{X}) = \mathbf{v}^T \Sigma \mathbf{v} \leq \lambda_{max}  \mathbf{v}^T  \mathbf{v} \leq \lambda_{max},
\end{aligned}
\end{equation*}
for any $\mathbf{v} \in {\cal S}_d$ because $\mathbf{v}^T \mathbf{v} \leq \mathbf{v}^T \mathbf{1} = 1$. Thus, for an accuracy of $\epsilon > 0$ for the MC approximation, it is sufficient to generate $M$ MC samples such that 
\begin{equation*}
\begin{aligned}
\sqrt{\frac{2 \ln \ln M}{M} \lambda_{max}} \leq \epsilon, \; \; i.e., \; \; \frac{M}{\ln \ln M} \geq \frac{2 \lambda_{max}}{\epsilon^2}
\end{aligned}
\end{equation*}
\end{proof}
Notice that $\Sigma$ and hence $\lambda_{max}$ can be easily estimated from the observed return values without any modeling assumption as long as $n>d$ , however sparse methods are necessary for large size portfolios when $n\leq d$. Next it can be shown that minimizing (\ref{eq:cvar_app}) is equivalent to minimizing
\begin{equation}
\begin{aligned}
F_{\alpha}(\mathbf{v}, \gamma) = \gamma + \frac{1}{\alpha M} \sum_{k=1}^M z_k \; \; s.t. \; \; z_k \geq 0,  z_k + \mathbf{v}^T \mathbf{x}_k + \gamma \geq 0
\end{aligned}
\label{eq:cvar_app2}
\end{equation}
Thus, along with the linear constraints on the weights $\mathbf{v}$, it can be formulated as a linear programming problem and can be solved using standard convex optimization methods. Conveniently, $\tt R$ function $\tt BDportfolio\_optim$ within the package $\tt PortfolioOptim$ can be used for this purpose.
Following the algorithm in \citet{semenov2017portfolio}, simulated return values can be obtained using estimated ECBC for portfolio optimization. The complete algorithm is summarized below:

\textbf{Step 1.} Transform assets' historical data $X_{tj}$ to pseudo-observations $U_{tj}$ and estimate copula using our proposed method.
\begin{equation*}
U_{tj} = F_{Tj}^{*}(X_{tj}), \; \; t=1,\ldots, T, j \in \{1, \dots, d \},
\end{equation*}
\begin{equation*}
F_{Tj}^{*}(x_j) =  \frac{1}{T+1} \sum_{t=1}^{T} I(X_{tj} \leq x_j), \; \; j \in \{1, \dots, d \}.
\end{equation*}

\textbf{Step 2.} Generate a sample of pseudo-observations $(U^*_{k1}, \ldots, U^*_{kd}), k = 1, \ldots, M$ from the estimated ECBC using empirical Bayes method and transform simulated pseudo-observations to univariate quantiles.
\begin{equation*}
X^*_{kj} = F_{Tj}^{* -1}(U^*_{kj}), \; \;k = 1, \ldots, M, j \in \{1, \dots, d \}.
\end{equation*}

\textbf{Step 3.} Calculate optimal weights $v_j, j \in \{1, \dots, d \}$ using simulated data $(X^*_{k1}, \ldots, X^*_{kd})$, $k = 1, \ldots, M$, and the corresponding VaR and CVaR, which are by-products of the portfolio optimization, by solving the linear programming problem given in (\ref{eq:cvar_app2}).

Our copula estimator is useful to find optimal weights and estimate risk measures as sampling from the estimated copula is straightforward. Considering the Bernstein copula density given in (\ref{eq:density1}) and (\ref{eq:density2}), we can obtain samples $(U_1, \ldots U_d) \sim {C}_{m}$ as follows
\begin{equation*}
\begin{aligned}
(k_1, \ldots, k_d) &\sim \tilde{w}_{k_1,\ldots, k_d}, \; \; k_j \in \{0,\ldots,m_j-1 \}, j \in \{1, \dots, d \}\\
U_j &\sim \text{Beta}(k_j + 1, m_j - k_j) , \; \; j \in \{1, \dots, d \}\\
\end{aligned}
\end{equation*}

As an example with moderately large dimension, we investigate the time series of daily closing stock prices of 10 top Nasdaq companies: AMZN, FB, GOOGL, AAPL, MSFT, INTC, CSCO, NFLX, CMCSA, and ADBE for the time period from January 1, 2018, to December 31, 2019. This data set consists of 502 observations and can be obtained using $\tt R$ package $\tt quantmod$. 

Suppose we want to find an optimal portfolio of stocks above that minimizes the expected shortfall of the portfolio. First, we convert the price series $P_{tj}$ to log-returns $X_{tj}$
\begin{equation*}
\begin{aligned}
X_{tj} = \ln \frac{P_{tj}}{P_{(t-1)j}}, \; \; t=1,\ldots, T, j \in \{1, \dots, d \},
\end{aligned}
\end{equation*}
resulting in $T = 501$ log-return values for $d=10$ assets. Then we follow Steps 1-3 as above to obtain the optimal portfolio weights and the corresponding VaR and CVaR. Similar to \citet{semenov2017portfolio}, we set the minimum weight to be limited by $v_j\geq 0.01, j \in \{1, \dots, d \}$ to avoid corner portfolio cases.

In Step 1, the posterior mode estimator of the degrees of the proposed ECBC are $(209, 206, 208, 206, 208, 209, 211,$ $210, 209,$ $208)$. The largest eigenvalue of the covariance matrix is $\lambda_{max} \approx 2 \times 10^{-3}$, so we are able to find the value of $M$ that is sufficient for a given accuracy $\epsilon$ from the relationship in (\ref{eq:lil_bound}).

We set $M = 10000$ (adequate for an accuracy $\epsilon \approx 9 \times 10^{-4}$) and repeat Step 2-3 $N = 100$ times to quantify estimation uncertainty. For each replicate we conduct portfolio optimization at the level of $\alpha\in \{0.10, 0.05, 0.01\}$ as popularly used. As a result, we are able to obtain the distribution of optimal weights (Fig. \ref{fig:10dist}) and risk measures (\ref{fig:10dist2}) using simulated data from the estimated copula.

From the boxplots of optimal weights in Fig. \ref{fig:10dist} we can see that CMCSA has a much higher weight than the other stocks in the mean-CVaR optimal portfolio across different levels. Also, by applying mean-CVaR portfolio optimization to the historical log-return data $X_{tj}$, we can get estimates of optimal weights and risk measures as well. In Fig. \ref{fig:10dist2}, the dashed lines indicate empirical estimates of risk measures using historical data.

\begin{figure}
    \centering
        \begin{subfigure}[b]{0.7\textwidth}
        \includegraphics[width=\textwidth]{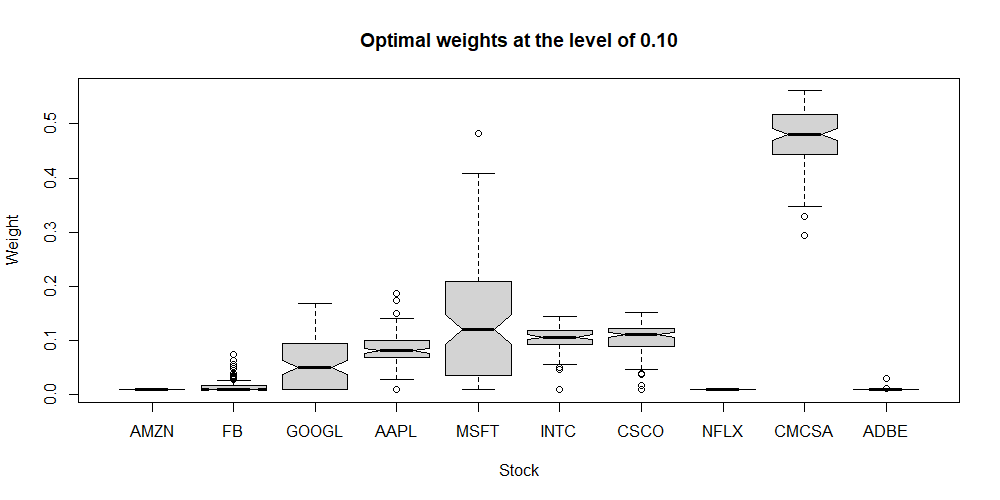}

    \end{subfigure}
             \begin{subfigure}[b]{0.7\textwidth}
        \includegraphics[width=\textwidth]{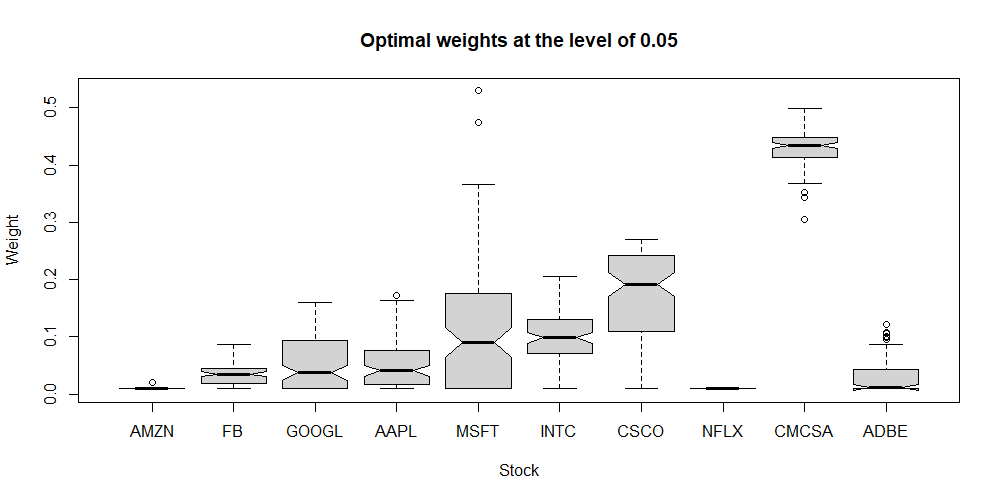}
 
    \end{subfigure}
             \begin{subfigure}[b]{0.7\textwidth}
        \includegraphics[width=\textwidth]{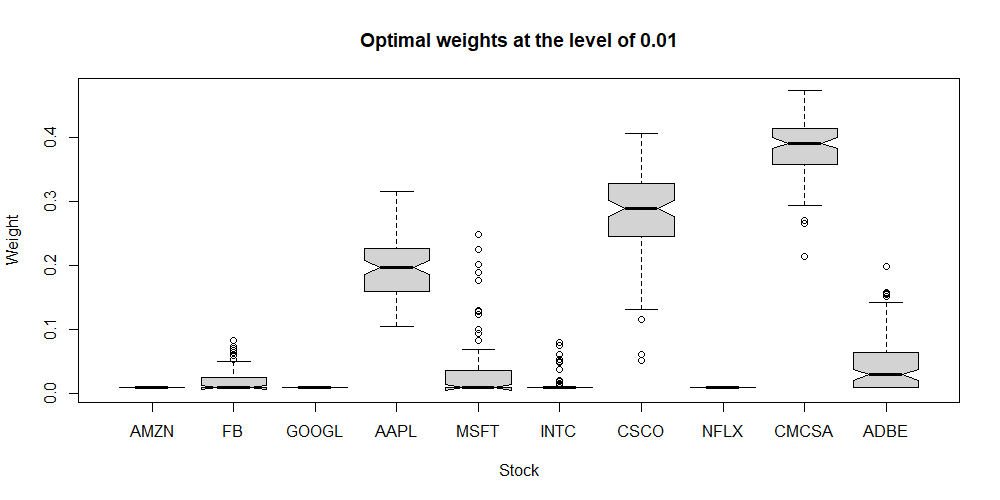}

    \end{subfigure}
    
    \caption{Distribution of optimal weights obtained from simulated data using the estimated copula at the level of $0.10, 0.05$, and $0.01$, for a portfolio of $d=10$ stocks.}\label{fig:10dist}
\end{figure}

\begin{figure}
    \centering
        \begin{subfigure}[b]{0.40\textwidth}
        \includegraphics[width=\textwidth]{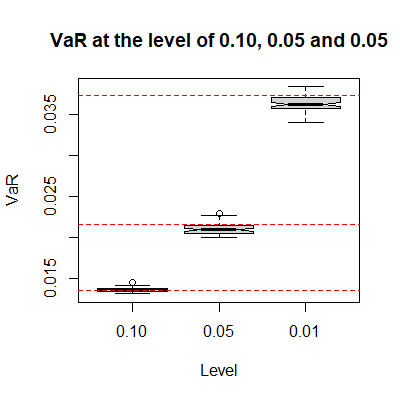}
        \caption{}
    \end{subfigure}
             \begin{subfigure}[b]{0.40\textwidth}
        \includegraphics[width=\textwidth]{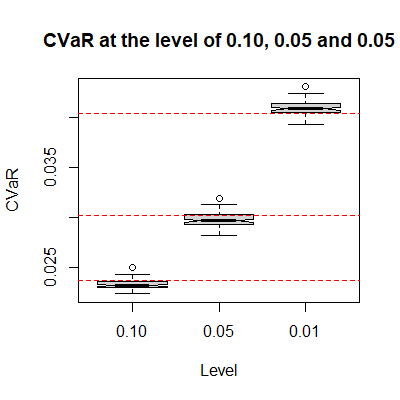}
        \caption{}
    \end{subfigure}

    \caption{Distribution of VaR and CVaR obtained from simulated data using the estimated copula at the level of $0.10, 0.05$, and $0.01$ for a portfolio of $d = 10$ stocks. Dashed lines indicate empirical estimates using historical data.}\label{fig:10dist2}
\end{figure}

We can see from the plots in Fig. \ref{fig:10dist2} that the estimated risk measures from two different methods seem to be fairly close. However, we are able to quantify the uncertainty for all the estimates by repeatedly sampling from the estimated copula. \citet{semenov2017portfolio} conduct a similar stability study to report the means and SDs of VaR and CVaR, but they used predetermined weights based on historical data and didn't report the distribution of optimal weights obtained from simulated data. Our copula estimator has shown good performance for relatively small samples and operationally we can generate as many samples as we want from the estimated copula, thus the copula-based method would be more reliable when there is not sufficient historical data. Besides, compared to the empirical estimates, it is possible to estimate VaR and CVaR for much smaller values of levels using the copula-based method.

\section{Concluding Remarks}
\label{sec:conc}

In this paper, we propose the empirical checkerboard Bernstein copula, which is a nonparametric multivariate copula estimator. It can be considered as an advancement of the empirical Bernstein copula since it is a valid copula with any polynomial degrees for any sample size. For automatic data-dependent dimension-varying degree selections, we further developed an empirical Bayesian method that has been shown to be practically useful. While the proposed copula estimator is shown to be large-sample consistent, it also has a good finite-sample performance. Moreover, it has a beneficial effect on measuring the strength of dependence for large dimensions because the estimates derived from the proposed copula are always within the proper range. 

As sampling from the estimated copula is quite straightforward, it is applicable to portfolio optimization and risk measurement where estimation is often done with simulations generated from copulas. We investigate the number of simulations that are good enough to achieve any given accuracy, which has been apparently out of reach in the literature. Furthermore, we are able to provide uncertainty quantification for all the estimates in portfolio risk management.

Under the hierarchical structure of the proposed empirical Bayes model, MCMC methods have been shown to work reasonably fast for relatively large dimensions ($d \leq 50$) with a moderate sample size ($n=100$). To speed up the MCMC methods for very large sample sizes, it would be of interest to explore some scalable MCMC methods such as divide-and-conquer approaches and subsampling approaches (see, e.g.,  \citet{quiroz2018speeding}, \citet{robert2018accelerating}, etc.). The codes (written using R software) to implement the procedure is available upon request from the first author and could be made available on a GitHub page following the publication of this paper.

\bibliographystyle{myjmva}

\bibliography{main}

\end{document}